\documentclass{article}
\usepackage{authblk}
\usepackage[T1]{fontenc}
\usepackage{microtype}
\usepackage{graphicx} 
\usepackage{subcaption}
\usepackage{amsthm} 
\usepackage{amssymb}
\usepackage{booktabs}
\usepackage[margin=2.5cm]{geometry}
\usepackage{nicefrac, xfrac} \usepackage{algorithm}
\usepackage[noend]{algpseudocode}
\usepackage{amsmath}
\usepackage{hyperref} 
\hypersetup{
    colorlinks=true,
    linkcolor=blue,
    citecolor=blue}
\usepackage[capitalise]{cleveref}
\usepackage{comment}
\usepackage{xcolor}
\usepackage{todo}
\usepackage{tikz}
\usepackage{tikzpeople}
\usetikzlibrary{decorations.pathreplacing,positioning,calc,backgrounds,intersections,shapes.geometric,arrows.meta,fit,patterns}
\tikzset{
    >=latex
}

\usepackage{mathtools}

\DeclareMathOperator{\TAG}{TAG}
\DeclareMathOperator{\VRFY}{VRFY}
\DeclareMathOperator{\ENC}{ENC}
\DeclareMathOperator{\COMP}{COMP}

\usepackage[normalem]{ulem}

\newcommand{\calH}{\mathcal{H}}
\newcommand{\calK}{\mathcal{K}}
\newcommand{\calM}{\mathcal{M}}
\newcommand{\calT}{\mathcal{T}}
\newcommand{\calZ}{\mathcal{Z}}
\newcommand{\calA}{\mathcal{A}}
\newcommand{\calC}{\mathcal{C}}
\newcommand{\calS}{\mathcal{S}}
\newcommand{\calF}{\mathcal{F}}

\newtheorem{theorem}{Theorem}[]
\newtheorem{definition}{Definition}[]

\theoremstyle{remark}

\newtheorem{remark}{Remark}

\date{}
\begin{document}

\author[1]{Claudia De Lazzari\thanks{claudia.delazzari@qticompany.com}}
\author[2]{Francesco Stocco}
\author[2]{Edoardo Signorini}
\author[3]{Giacomo Fregona}
\author[1]{Fernando Chirici}
\author[1]{Damiano Giani}
\author[1]{Tommaso Occhipinti}
\author[2]{Guglielmo Morgari}
\author[1,4]{Alessandro Zavatta}
\author[1,5]{Davide Bacco}

\affil[1]{QTI S.r.l. --- Largo Enrico Fermi 6, 50125, Florence, Italy}
\affil[2]{Telsy S.p.A. --- Corso Svizzera 185, 10149, Turin, Italy}
\affil[3]{Department of Computer Science, University of Copenhagen ---	Universitetsparken 1 DK-2100 Copenhagen Ø}
\affil[4]{National Institute of Optics
National Research Council (INO--CNR) --- Largo Enrico Fermi 6, 50125 Florence, Italy
}
\affil[5]{Department of Physics and Astronomy, University of Florence --- 50019, Florence, Italy
}

\title{Send the Key in Cleartext: Halving Key Consumption while Preserving Unconditional Security in QKD Authentication}

\maketitle

\begin{abstract}
	Quantum Key Distribution (QKD) protocols require Information-Theoretically Secure (ITS) authentication of the classical channel to preserve the unconditional security of the distilled key.
	Standard ITS schemes are based on one-time keys: once a key is used to authenticate a message, it must be discarded.
	Since QKD requires mutual authentication, two independent one-time keys are typically consumed per round, imposing a non-trivial overhead on the net secure key rate.
	In this work, we present the \emph{authentication-with-response} scheme, a novel ITS authentication scheme based on $\varepsilon$-Almost Strongly Universal\textsubscript{2} ($\varepsilon$-ASU\textsubscript{2}) functions, whose IT security can be established in the Universal Composability (UC) framework.
	The scheme achieves mutual authentication consuming a single one-time key per QKD round, halving key consumption compared to the state-of-the-art.\\

	Keywords: Information-theoretic secure authentication, Quantum key distribution,  Partially known keys, Almost strong universal functions, Universal composability
\end{abstract}

\section{Introduction}
Quantum Key Distribution (QKD) allows distant parties to establish an unconditionally secure shared key by exploiting the laws of quantum mechanics~\cite{BENNETT20147,Ekert1992}.
In QKD protocols, the exchange of quantum states through a quantum channel is followed by classical post-processing carried out over a public classical channel, in order to distill the final secure key.
While the quantum channel is treated as untrusted, the classical channel must be \emph{authenticated}: both parties must be able to verify the identity of their counterpart and the integrity of every exchanged message, except with an explicitly quantified failure probability.

To highlight the crucial role played by authentication, observe that in the absence of an authenticated classical channel, any communication is vulnerable to Man-In-The-Middle (MITM) attacks.
In a MITM attack, the adversary interposes itself between two legitimate parties, usually referred to as Alice and Bob, impersonating one to the other throughout the communication and potentially reading, replacing, or injecting messages undetected.
Authentication is therefore foundational to secure communication in the presence of an active adversary and is considered at least as critical as confidentiality.

An authenticated channel can be built using various cryptographic techniques.
In the symmetric-key setting, the standard tool is a Message Authentication Code (MAC), which derives a short tag from the message and a shared secret key: an adversary who does not know the key cannot forge a valid tag.
In the public-key (or asymmetric-key) setting, the relevant primitive is a Digital Signature, which additionally provides the stronger property of non-repudiation: the signer cannot later deny having produced the signature.

Classical cryptography is typically based on \emph{computational security}, i.e., it relies on the assumption that no efficient (polynomial-time) algorithm exists for certain mathematical problems, and thus offers security only against computationally bounded adversaries.

Since QKD targets key establishment within the stronger \emph{unconditional security} paradigm, in which security must hold against any admissible quantum and classical attack, in particular any adversary with arbitrary computational power, Information-Theoretic Secure (ITS) authentication is required. In an ITS scheme, no adversary, regardless of its computational resources, should have non-negligible advantage over random guessing of a valid tag.

The foundational framework of ITS authentication was laid by Wegman and Carter~\cite{WC81}, who introduced an authentication scheme based on sets of Almost Strongly Universal (ASU) functions and proved its unconditional security.
In their scheme, a shared secret key selects an element from a publicly known set of ASU\textsubscript{2} functions; the tag is then the function evaluated on the message.

The original Wegman-Carter analysis assumes a \emph{perfect} shared key, i.e., one that is uniformly distributed over the key space. In the QKD  context, however, the authentication key is itself produced by the QKD system: it is drawn from the pool of distilled key material, which is only $\varepsilon'$-close to uniform (in trace distance) due to residual information leakage to an eavesdropper. Applying the classical Wegman-Carter result directly would therefore be incorrect, as its security guarantees are conditioned on a perfectly uniform key.
The extension of the theory to such \emph{partially known keys} is due to Abidin and Larsson~\cite{AL14}, who show that Wegman-Carter authentication with an $\varepsilon'$-perfect key and an $\varepsilon$-ASU\textsubscript{2} set remains secure with failure probability $\varepsilon+\varepsilon'$. Their result is established within the Universal Composability (UC) framework~\cite{CAN01}, which ensures that the security guarantee is preserved when the authentication protocol is combined with other protocol components, in particular with the quantum key establishment itself. Concretely, security in the UC framework is assessed via a distinguisher-based definition: the authentication protocol is $\varepsilon$-secure if no environment can distinguish the real implementation from the ideal authenticated channel, except with a small failure probability.
The present work is positioned in this line of research.

While early QKD works typically assumed an authenticated classical channel as a black box, sustained attention to the concrete cost and design of authentication has emerged more recently, driven by practical deployments of QKD systems~\cite{SCA09,FUNG10}. Two topics have become particularly relevant: the impact of authentication on the net secure key rate, and the bootstrapping problem of the very first round.
The former concerns the fact that authentication keys are themselves secret material, so their consumption directly reduces the amount of key available for other uses. The latter concerns the question of how the very first QKD session can be authenticated before any QKD-generated key is available, typically requiring a pre-shared key of sufficient quality.
This work addresses the former but does not enter into the latter, although we regard it as important and therefore mention it here.

During classical communication of the QKD protocol, authentication must be mutual: Alice and Bob authenticate to each other, w.l.o.g.\ first Alice to Bob, then Bob to Alice, so that both parties can verify the integrity of the full round's classical transcript.
In both directions, authentication is realized by transmitting a tag computed under a shared secret key, pre-shared for the first round, or drawn from the QKD system's key pool in subsequent rounds.
Consequently, two independent one-time authentication keys are consumed per QKD round, one per direction. Since each key can be used only once to maintain ITS security (reuse would allow an adversary to correlate multiple tag-message pairs and undermine the ASU\textsubscript{2} guarantee), repeated QKD rounds impose a non-trivial authentication overhead that reduces the net secure key rate.
To mitigate this issue, in modern practice, individual classical messages are not authenticated one by one; instead, both parties accumulate their messages and perform a single end-of-round authentication.
Also, mechanisms for partial key recycling have been proposed, see for example Ref.~\cite{PORTMANN14}.

\paragraph{Related works and our contributions.}
Among the relevant related works, Kiktenko \emph{et al.}~\cite{KIK20} propose a mutual authentication method that spreads authentication over two consecutive QKD rounds: in the $i$-th round Alice authenticates to Bob, and in round $(i+1)$-st Bob authenticates to Alice the full transcript of round $i$ and $i+1$ in an alternate fashion. They prove IT security and also discuss universal function classes beyond ASU\textsubscript{2}. The approach halves the average authentication key consumption per round. Its main drawback is that a failure in round $i+1$ invalidates round $i$ retroactively, even if the first verification had succeeded, introducing a coupling between consecutive rounds.

This observation motivates the search for an authentication mechanism that is mutual, parsimonious in terms of key consumption, and that preserves the independence of consecutive QKD rounds.
The present work introduces a novel scheme, the \emph{authentication-with-response} scheme~\cite{footnote1}, which achieves all three goals.
The high-level idea is the following: Alice authenticates to Bob by computing a one-time MAC tag under a shared key $k$; upon successful verification, Bob authenticates to Alice by revealing $k$ itself. Since $k$ is one-time and has already been consumed, its disclosure is harmless to security. The proposed scheme consumes a single authentication key per round for mutual authentication, halving key consumption compared to standard approaches while maintaining ITS security within the UC framework.
We describe the scheme in full, introduce the necessary mathematical background, and establish its security via a distinguisher-based proof within the UC framework, following the techniques of Ref.~\cite{AL14}.

\paragraph{Structure of the paper.}
The manuscript is organized as follows. \Cref{sec:pre} introduces the mathematical tools used throughout the paper and discusses the security framework in which ITS authentication schemes are typically formulated;
\cref{sec:aws} describes the proposed {authentication-with-response} scheme and proves its security;
\cref{app:mult:round} provides additional insights into the security of the proposed scheme by comparing it with state-of-the-art solutions.

\section{Preliminaries}\label{sec:pre}
This section presents the notation used throughout the manuscript (\cref{sec:notation}); the $\varepsilon$-Almost Strongly Universal\textsubscript{2} functions (\cref{sec:ASU}); the trace distance between probability distributions, our principal metric for studying authentication with partially known keys (\cref{sec:partiallykk}); and the distinguisher-based method from UC framework (\cref{sec:UC}), which is the technique we adopt to establish our main result.

\subsection{Notation}\label{sec:notation}
Throughout this work, for any set $\calS$, its cardinality is denoted by $|\calS|$.
Consider a discrete random variables $X : \Omega \to \mathcal{X}$ defined on a discrete probability space $(\Omega, \mathcal{P}(\Omega), P)$ and a discrete measurable space $(\mathcal{X}, \mathcal{P}(\mathcal{X}))$, where $\mathcal{P}(\cdot)$ denotes the power set.
The probability distribution of $X$ is written $P_X$, so that $P_X(x) = \Pr[X = x]$ for each $x \in \mathcal{X}$, and $\Pr[X \in \mathcal{X}'] = \sum_{x \in \mathcal{X}'} P_X(x)$ for any $\mathcal{X}' \subseteq \mathcal{X}$.
Given a further random variable $Y : \Omega \to \mathcal{Y}$ on a discrete measurable space $(\mathcal{Y}, \mathcal{P}(\mathcal{Y}))$, the joint probability distribution (resp.\ conditional distribution) is denoted by $P_{XY}(x,y) = \Pr[X = x, Y = y]$ (resp. $P_{X \mid Y}(x \mid y) = \Pr[X = x \mid Y = y]$) for $x \in \mathcal{X}$ and $y \in \mathcal{Y}$.
Finally, for a predicate $\mathbf{eq}$ over $\mathcal{X}$, we write $\Pr[\mathbf{eq}(X)]$ as shorthand for $\Pr\bigl[X \in \{x \in \mathcal{X} \mid \mathbf{eq}(x)\}\bigr]$.

\subsection{Authentication with Almost Strong Universal functions}\label{sec:ASU}
In this section, the authentication scheme of Wegman and Carter~\cite{CW79,WC81}, which achieves information-theoretic security, is recalled. Their work revolves around the idea of \emph{$\varepsilon$-Almost Strongly Universal$_2$} functions, which has been mathematically formalized later by Stinson~\cite{STINSON94}. This particular strategy is known to be widely used in real QKD systems~\cite{AL14, FREG24}.

The main mathematical object, \emph{$\varepsilon$-Almost Strongly Universal\textsubscript{2}} functions, are therefore introduced.
\begin{definition}[$\varepsilon$-ASU\textsubscript{2} functions]\label{def:epsASU}
	Let $\calM$ and $\calT$ be finite sets, let $\calH$ be a set of functions from $\calM$ to $\calT$ and let $\varepsilon$ be a positive real number. The set $\calH$ is called $\varepsilon$-Almost Strongly Universal$_2$ if the following conditions hold:
	\begin{enumerate}
		\item\label{ia}for any $m_1\in\calM$ and $t_1\in\calT$:
		      \begin{equation*}
			      \dfrac{|\{h\in\calH\mid h(m_1)=t_1\}|}{|\calH|}=\dfrac{1}{|\calT|};
		      \end{equation*}
		\item\label{sa} let $\calH_1:=\{h\in\calH\mid h(m_1)=t_1\}$, for any $m_2\in\calM\backslash \{m_1\}$ and $t_2\in\calT$:
		      \begin{equation*}
			      \dfrac{|\{h\in\calH_{1}\mid h(m_2)=t_2\}|}{|\calH_{1}|}\leq\varepsilon.
		      \end{equation*}
	\end{enumerate}
\end{definition}

Now, let $\calM$ be the set of all possible \emph{messages}, $\calT$ the set of all possible \emph{tags} and $\calK$ the set of \emph{keys}.
Each key $k\in \calK$ defines a function of the $\varepsilon$-ASU\textsubscript{2} set $\calH=\{ h_{k}:\calM \to \calT\}_{k\in\calK}$.
That is, the secret key shared by two parties (Alice and Bob), is an element $k\in\calK$ that serves as an index selecting $h_{k}\in\calH$; in particular $|\calH|=|\calK|$.
Moreover, let $K$ a uniformly distributed random variable over $\calK$; Conditions~\ref{ia} and~\ref{sa} of \cref{def:epsASU} can be rewritten in terms of probabilities as follows:
\begin{enumerate}
	\item\label{iak}for any $m_1\in\calM$ and $t_1\in\calT$:
	      \begin{equation*}
		      \Pr[h_K(m_1)=t_1]=\dfrac{1}{|\calT|};
	      \end{equation*}
	\item\label{sak} for any $m_2\in\calM\backslash \{m_1\}$ and $t_1, t_2\in\calT$:
	      \begin{equation*}
		      \Pr[h_K(m_2)=t_2 \mid h_K(m_1)=t_1]\leq\varepsilon.
	      \end{equation*}
\end{enumerate}

Given the use of this set of functions\footnote{Notice that from a practical point of view it is desirable that the length of the key (as well as the length of the tag) is much smaller than the message length, i.e. $\log|\calM|\gg\log|\calK|$ and $\log|\calM|\gg\log|\calT|$, where $\log$ denotes the binary logarithm. This is a common requirement in cryptographic applications.}, the basic authentication scheme proceeds as follows.
If Alice wants Bob to receive a message $m$ through an insecure channel with (almost) no doubts about authenticity, she computes the MAC $h_{k}(m)=t$ and sends the pair  $(m,t)$. To verify authenticity, Bob checks if $h_k(m)=t$ since he knows the secret key $k$. If the check succeeds he keeps the message $m$, otherwise the communication is aborted.
To break the system an attacker, Eve, must be able to produce a correct tag without the information $k$. The eavesdropper has at least two strategies that in the literature are called \emph{impersonation attack} and \emph{substitution attack}, and
each of the two $\varepsilon$-ASU\textsubscript{2} defining conditions directly corresponds to one of them. In the first case, Eve wants to impersonate Alice to Bob (or vice versa) and tries to trick Bob (or Alice) by guessing a tag $t_{E}$ for her message $m_{E}$.
Considering Condition~\ref{iak},
Eve's success probability is just $1/|\calT|$, i.e., the tag is completely random to Eve.
The substitution attack may be more incisive.
Eve detects a pair  $(m_A,t_A)$ sent by Alice to Bob (or vice versa) and substitutes it with $(m_E,h(m_E))$ choosing a random element $h\in\{h_{k}\in\calH\mid h_{k}(m_A)=t_A\}$\footnote{Obtaining this set might be quite arduous from a computational point of view. However,
	Eve has unlimited computational power in the ITS model.}. According to Condition~\ref{sak},
the success probability of this strategy is bounded by $\varepsilon$.

Finally, observe that the $\varepsilon$-ASU\textsubscript{2} conditions, by themselves, provide no guarantees when multiple message-tag pairs, obtained with the same function, are known to Eve.
More formally, let $n>1$ be an integer,
let $(m_1,h_k(m_1)),\dots,(m_n,h_k(m_n))$ be different message-tag pairs corresponding to the same secret key $k$ and define $\calH_{n}:=\{h\in\calH\mid h(m_{i})=h_{k}(m_i),\,\forall  i\leq n\}$.
To the best of our knowledge, for any message $m_{n+1}\neq m_{i}$ for all $i\leq n$, there is, a priori, no non-trivial bound on:
\begin{equation*}
	\dfrac{|\{h\in\calH_n\mid h(m_{n+1})=h_{k}(m_{n+1})\}|}{|\calH_{n}|}.
\end{equation*}
For this reason, at the end of each authentication step, the corresponding key $k$ is discarded, and this ITS MAC construction is also called one-time MAC.\@

\subsection{Authentication with Partially Known Keys}\label{sec:partiallykk}
A QKD system operates as a continuous key generator, therefore the number of needed authentication instances is neither predetermined nor bounded.
Each QKD round produces a certain amount of secure key material which contributes to filling a pool of keys for later uses, with a designated fraction reserved to authenticate messages in later QKD rounds.
In the case of authentication with $\varepsilon$-ASU\textsubscript{2}, the keys identify new elements of the set.

At the beginning of the first QKD round, the authentication key is completely random to an attacker, as practical QKD systems are usually equipped with a pre-shared key.
Thereafter, an eavesdropper may gain information on the key distributed during the round.
Therefore, from the second round onward, it is required to account for Eve's partial knowledge of the keys produced by the QKD system, part of which are used to select elements of the $\varepsilon$-ASU\textsubscript{2} set for authentication purposes.

This motivates introducing a distance between a ``perfect'' key (i.e., uniformly distributed over the key space) and the key generated by the QKD system. Information leakage is quantified by the \emph{trace distance} between the probability distributions.
\begin{definition}
	Let $(\Omega,\calF, P)$, $(\Omega,\calF, Q)$ be two probability spaces and let $X:\Omega\rightarrow\mathcal{X}$ be a random variable. The \emph{trace distance} or \emph{statistical distance} between two probability distributions $P_{X}$ and $Q_{X}$ is defined as:
	\begin{equation}
		\delta(P_{X},Q_{X})=:\frac{1}{2}\sum_{x\in\mathcal{X}}|P_{X}(x)-Q_{X}(x)|.
	\end{equation}

	Moreover, a random variable $K:\Omega\rightarrow\calK$ over the probability space  $(\Omega,\calF, P)$ is called \emph{perfect} if it is uniformly distributed, i.e., any $k\in\calK$ is such that $P_{K}(k)=1/|\calK|$. More generally, a random variable $K$ is called $\varepsilon'$-\emph{perfect} if its distribution has an $\varepsilon'$ trace distance to the uniform.
\end{definition}
Therefore, if the trace distance between the distribution of the key produced by the real QKD system and the uniform distribution is at most $\varepsilon'$, then the key used for authentication is $\varepsilon'$-perfect\footnote{In the literature, reported values of $\varepsilon'$ range from $\sim10^{-9}$ to $\sim10^{-15}$}.
In the UC framework, it is standard to choose the ASU\textsubscript{2} parameter $\varepsilon$ at least two orders of magnitude smaller than $\varepsilon'$ when this $\varepsilon$-perfect key is used for authentication~\cite{KIK20,FREG24}.

\subsection{Universal Composability Framework}\label{sec:UC}
The full-fledged security of a cryptographic protocol does not imply that it can be arbitrarily combined with other protocols while maintaining the same level of security. The Universal Composability (UC) framework~\cite{CAN01, MAU11} is a strong model in which protocols (resources) are proven to be secure even when arbitrarily combined with each other. Security is formulated in terms of indistinguishability, meaning that a resource is called $\varepsilon$-secure if a distinguisher is not able to determine whether it is interacting with the \emph{real implementation} or the \emph{ideal functionality} of the resource, except with probability at most $\frac{1}{2}+\frac{1}{2}\varepsilon$.

In the context of MAC, the ideal functionality is represented by an authenticated channel which connects Alice and Bob, as shown in \cref{fig:ideal_gen}. If a message $m$, sent from Alice to Bob, is modified during the communication, the authenticated channel delivers to Bob an error value $\perp$. On the other hand, if no message alteration occurs during the transmission, Bob receives the expected value $m$. In other words, messages exchanged through an authenticated channel are either authentic or blocked. Therefore, an attacker (Eve) cannot modify or substitute messages without being detected.

\begin{center}
	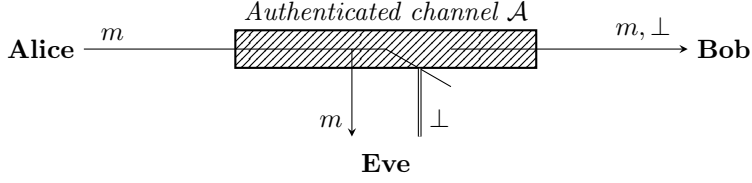
\begin{figure}[tb]
		\centering
		\begin{tikzpicture}[
				>=stealth,
				resource/.style={draw, thick, fill=white, rectangle, minimum height=0.5cm, minimum width=4cm},
				eve/.style={criminal, female, monitor}]

			\node[resource, pattern=north east lines, label=above:{\itshape Authenticated channel $\calA$}] (channel) at (0,0) {};
\node[left = 2cm of channel, font=\bfseries] (alice) {Alice};
\node[right=2cm of channel, font=\bfseries] (bob) {Bob};
\node[below = 1cm  of channel, font=\bfseries] (eve) {Eve};

			\draw[->] (alice.east) node[above right = 0cm and 0.1cm]{$m$} -| ([yshift=0.1cm]eve.north west) node[above left]{$m$};
			\draw[] (alice.east -| eve.north west) -- (channel.center);
			\draw[name path=diagonal] (channel.center) -- +(-30:1) coordinate (ediagonal) ;
			\path[name path=line1] (eve.north east) -- (eve.north east |- channel.center);
			\draw[name intersections={of=diagonal and line1,by={Int1}}, double] ([yshift=0.1cm]eve.north east) node[above right]{$\perp$} -- (Int1);
			\draw[->] (ediagonal |- channel.center) -- (bob.west) node[above left = 0cm and 0.1cm]{$m, {\perp}$};
		\end{tikzpicture}
		\caption{Ideal authentication functionality: Alice sends a message $m$ through the authenticated channel $\calA$. Depending on Eve’s actions, either the message is delivered intact (Bob receives $m$) or the authentication protocol fails (Bob receives $\perp$). }\label{fig:ideal_gen}
	\end{figure}
\end{center}

In the real implementation, represented in \cref{fig:real_gen}, an insecure channel is combined with a secret key source to emulate an authenticated channel. Using the Wegman-Carter scheme, this can be modeled as follows. A key source guarantees that Alice and Bob obtain a common secret $k$ which allows them to select an element $h_k$ from a public family of $\varepsilon$-ASU\textsubscript{2} functions. A tag $h_k(m)$ is then computed from a message $m$ and the pair $(m,h_k(m))$ is sent along the insecure channel. Once the pair is received, a verification step checks whether the incoming pair $(m',t')$ satisfies $h_k(m')=t'$. If such verification succeeds then the message $m'$ is accepted, otherwise it is rejected. Therefore, an attacker is not able to modify or substitute an incoming message, while remaining undetected, unless it guesses the correct tag.

\begin{center}
	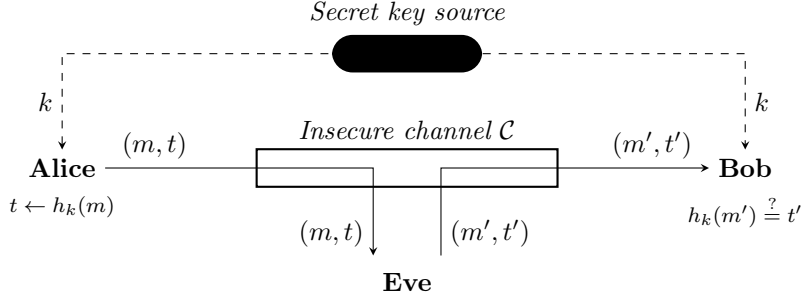
\begin{figure}[tbh]
		\centering
		\begin{tikzpicture}[
				>=stealth,
				resource/.style={draw, thick, fill=white, rectangle, minimum height=0.5cm, minimum width=4cm},
				eve/.style={criminal, female, monitor}]

			\node[resource, label=above:{\itshape Insecure channel $\mathcal{C}$}] (channel)  at (0,0) {};
\node[left = 2cm of channel, font=\bfseries] (alice) {Alice};
\node[right=2cm of channel, font=\bfseries] (bob) {Bob};
\node[below = 1cm  of channel, font=\bfseries] (eve) {Eve};

			\node[rectangle, rounded corners=2.75mm, minimum height=0.5cm, minimum width=2cm, fill=black, above=1cm of channel, label = above:{\itshape Secret key source}] (SKS) {};
			\draw[<-, dashed] (alice.north) |- (SKS.west) node[left, near start]{$k$};
			\draw[->, dashed] (SKS.east) -| (bob.north) node[right, near end]{$k$};

\node[anchor=north] at (alice.south) {\footnotesize$t \gets h_k(m)$};

\draw[->] (alice.east) node[above right = 0cm and 0.1cm]{$(m,t)$} -| ([yshift=0.1cm]eve.north west) node[above left]{$(m,t)$};
			\draw[ ->] ([yshift=0.1cm]eve.north east) node[above right]{$(m',t')$} |- (bob.west) node[above left = 0cm and 0.1cm]{$(m',t')$};

\node[anchor=north] at (bob.south) {\footnotesize$h_k(m') \stackrel{?}{=} t' $};

		\end{tikzpicture}
		\caption{Real authentication functionality: The secret key source delivers a shared secret key between Alice and Bob. Alice computes the tag $t=h_k(m)$ of a message $m$ and sends the pair $(m,h_k(m))$ along the insecure channel $\calC$. Eve intercepts it and resends $(m',t')$. Bob receives $(m',t')$ and verify if $h_k(m')=t'$. If such verification succeeds, then the message $m'$ is accepted, otherwise it is rejected.}\label{fig:real_gen}
	\end{figure}
\end{center}

A distinguisher is an environment which tests ideal functionality and real implementation and tries to guess which one it is interacting with. To do so, a distinguisher can play Alice's and Eve's roles sending a particular message $m$ and trying to substitute it with $m'$ when the transmission is running over the insecure channel. To formulate its guess, it can check also if Bob accepts or rejects $m'$.
When Wegman-Carter authentication with perfect keys with $\varepsilon$-ASU\textsubscript{2} function is used, no distinguisher can identify the real implementation over the ideal one with probability exceeding $\frac{1}{2}+\frac{1}{2}\varepsilon$; that is, the scheme is $\varepsilon$-secure; see for example Ref.~\cite{PORTMANN14}.

The main result linking authentication schemes based on $\varepsilon$-ASU\textsubscript{2} functions selected with imperfect keys and the UC framework is due to Abidin and Larsson~\cite{AL14}:
\begin{theorem}[Abidin, Larsson]
	Wegman and Carter's scheme based on $\varepsilon$-ASU\textsubscript{2} functions, employed with an $\varepsilon'$-perfect authentication key $k$, is $\varepsilon+\varepsilon'$-secure.
\end{theorem}

\section{Authentication-with-response Scheme}\label{sec:aws}
In practical QKD protocols, mutual authentication between the communicating parties is required. Therefore, in addition to the authentication procedure from Alice to Bob described above, an equivalent authentication step must be performed in the opposite direction.
Since authentication keys are one-time, each bidirectional classical communication consumes two independent authentication keys, one per direction, thereby potentially reducing the overall secret key rate.

To avoid massive authentication key consumption due to this one-time nature, Alice and Bob typically collect the entire set of classical messages exchanged within the QKD round and perform a single authentication step at its end, see e.g., Ref.~\cite{Cede08}. Their local reconstructed messages (denoted by $m_A$ and $m_B$) coincide in the absence of in-transit tampering. In this setting, Alice transmits only the authentication tag $t=h_k(m_A)$ at the end of the round, while Bob independently computes the expected tag $h_k(m_B)$ under the shared key $k$, and accepts if and only if the tags coincide. This arrangement is functionally equivalent to the one considered within the manuscript, i.e. Alice sends message \emph{and} tag, and Bob checks the pair using the common secret, and vice versa.
However, this method mitigates the key-consumption problem but does not completely solve it, since two one-time keys are still consumed for mutual authentication.

In this work, we propose a novel strategy, named \emph{authentication-with-response}, that preserves IT security while requiring only a \emph{single authentication key per round}, thus halving the key consumption associated with authentication.
The high-level idea behind the scheme is that once Alice has authenticated to Bob by computing a one-time MAC with key $k$, Bob verifies the received tag and authenticates to Alice by revealing $k$ itself.

The novelty of the proposed scheme lies in the response phase: after Alice authenticates to Bob using an $\varepsilon$-ASU\textsubscript{2} tag computed under the shared key $k$, and Bob verifies it, Bob authenticates to Alice by revealing $k$ itself.
This may sound surprising, since the scheme's security is based on the secrecy of $k$. The key observation is that $k$ is disclosed only after it has been fully consumed: as a one-time key it can never be reused, so its exposure at this point does not compromise security.

The scheme flow is as follows, and is schematically illustrated in \cref{fig:aut-res}. Let $\calM$ and $\calT$ be the sets of all possible messages and authentication tags of a certain fixed length, respectively. Let $\calK$ be a set of keys of a fixed length and $\calH = \{h_k: \calM \longrightarrow \calT\}_{k\in \calK}$ be a set of $\varepsilon$-ASU\textsubscript{2} functions.
Prior to authentication, Alice and Bob hold a pre-shared secret key $k \in \calK$ and have each accumulated a local transcript of the messages exchanged during the QKD round, denoted by $m_A$ and $m_B$ respectively, which may differ due to channel noise or adversarial intrusion.
The authentication-with-response scheme consists of the following steps:
\begin{enumerate}
	\item Let $m_A\in \calM$, the message collection owned by Alice, and $h_k\in \calH$. Alice computes the tag $t := h_k(m_A)\in \calT$ and sends $t\in\calT$ to Bob through the classical channel.
	\item Bob receives $t'\in \calT$, that may differ from $t$ due to channel noise or adversarial intrusion, and compares it with $h_k(m_B)\in \calT$, where $m_B$ is his message collection. The message is accepted if the received and computed tags are equal, i.e., if $t'=h_k(m_B)$.
	\item\label{resp} Bob defines the response $r$, setting it to $k$ in case the message has been accepted and to $\perp$ in the case the message has been refused (notice that the response value is not hidden from an attacker). Then $r$ is sent to Alice through the channel\footnote{Notice that the security proof of the scheme does not require the content of the failure response $\perp$ to be hidden. However, in some cases $\perp$ can be encoded as $k'\in\calK\backslash\{k\}$ such that $h_{k'}(m) = t$, e.g., using polynomial hashing~\cite{BO93}. In this way, an attacker is unable to distinguish between an accepting or rejecting response from Bob.}.
	\item\label{ver-res} Alice receives $r'$, which may differ from $r$ due to channel noise or adversarial intrusion, and compares it with $k$ to evaluate the incoming response. Equality represents success while inequality represents failure. Notice that any response other than $k$ or $\perp$ represents a failure for Alice.
\end{enumerate}

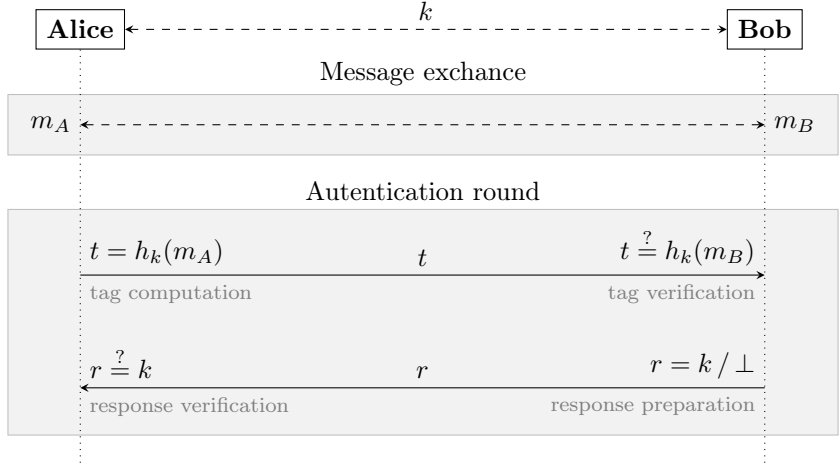
\begin{figure}[t]
	\centering
	\begin{tikzpicture}[
		>=stealth,
		party/.style      = {draw, fill=white, inner sep=4pt, font=\bfseries},
		note/.style       = {inner sep=2pt, align=left},
		annotation/.style = {note, font=\footnotesize, gray},
		msg/.style        = {->},
		phasebox/.style   = {fill=gray!10, draw=gray!50},
		inbox/.style      = {draw, rounded corners, fill=white, inner sep=2pt},
		num/.style        = {circle, draw, fill=white, inner sep=1.5pt, font=\footnotesize\bfseries},
	]

	\node[party] (A) {Alice};
	\node[party, right=8cm of A] (B) {Bob};

    \draw[<->, dashed] (A) -- node[above, midway] (key) {$k$} (B);

	\coordinate[below=1cm of A] (Ar0);
	\coordinate (Br0) at (B |- Ar0);
	\draw[<->, dashed] (Ar0) -- (Br0);
	\node[left=0cm of Ar0] (MA0) {$m_A$};
	\node[right=0cm of Br0] (MB0) {$m_B$};

	\coordinate[below=2cm of Ar0] (Ar1);
	\coordinate (Br1) at (B |- Ar1);
	\node[left=0cm of Ar1] (MA1) {$\phantom{m_A}$};
	\node[right=0cm of Br1] (MB1) {$\phantom{m_B}$};

	\node[note, anchor=west, above right=2pt of Ar1.east] (tcomp-f) {$t = h_k(m_A)$};
	\node[annotation, anchor=west, below right=2pt of Ar1] (tcomp-l) {tag computation};

	\node[note, anchor=east, above left=2pt of Br1.west] (tver-f) {$t \stackrel{?}{=} h_k(m_B)$};
	\node[annotation, below left=2pt of Br1] (tver-l) {tag verification};

	\draw[msg] (Ar1) -- node[above, midway]{$t$} (Br1);

	\coordinate[below=1.5cm of Ar1] (Ar2);
	\coordinate (Br2) at (B |- Ar2);

	\node[note, anchor=east, above left=2pt of Br2.west] (rprep-f) {$r = k \,/\, \bot$};
	\node[annotation, below left=2pt of Br2] (rprep-l) {response preparation};

	\node[note, anchor=west, above right=2pt of Ar2.east] (rver-f) {$r \stackrel{?}{=} k$};
	\node[annotation, anchor=west, below right=2pt of Ar2] (rver-l) {response verification};

	\node[above] (lbl-r) at ($(Ar2)!.5!(Br2)$) {$r$};
	\draw[msg] (Br2) -- (Ar2);

	\coordinate[below=1cm of Ar2] (A-bot);
	\coordinate (B-bot) at (B |- A-bot);
	\draw[dotted] (A.south) -- (A-bot);
	\draw[dotted] (B.south) -- (B-bot);

	\begin{scope}[on background layer]
		\node[phasebox, inner sep=5pt,
			fit=(Ar0)(Br0)(MA0)(MB0), label={[above]Message exchance}] {};
		\node[phasebox, inner sep=5pt,
			fit=(tcomp-f)(tcomp-l)
			(tver-f)(tver-l)
			(rprep-f)(rprep-l)
			(rver-f)(rver-l)
            (MA1)(MB1),
			label={[above] Autentication round}] {};
	\end{scope}
\end{tikzpicture}
 	\caption{High level representation of the Authentication-with-response scheme.}\label{fig:aut-res}
\end{figure}

In the following, the presented authentication-with-response scheme is proved to be secure in the UC framework along the lines of the proof in~\cite{AL14}.

\subsection{Scheme Functionalities}\label{sec:functs}
The security proof, provided in the UC framework, establishes that the real functionality is indistinguishable from its ideal counterpart, except with low probability. That is, after formalizing the real and ideal functionalities of the authentication-with-response scheme, we bound the trace distance between the probability distributions of the random variables that describe the two executions.

\begin{figure}[tb]
	\centering
	\begin{tikzpicture}[rect/.style={draw, minimum height=0.7cm}]
    \node[rect, minimum width=10cm] at (0,0) (Z) {\Large$\mathcal{Z}$};
    \node[rect, minimum width=1.5cm, above=4cm of Z.west, anchor=west] (alice) {Alice};
    \node[rect, minimum width=1.5cm, above=4cm of Z.east, anchor=east] (bob) {Bob};
    \coordinate (middle) at ($(alice)!0.5!(bob)$);
    \node[rect, minimum width=1.2cm, left=0cm of middle, xshift=-1.5em, yshift=-1em] (comp) {\small COMP};
    \node[rect, minimum width=1.2cm, above=0cm of comp.north, anchor=south east, yshift=0.5em] (tag) {\small TAG};
    \node[rect, minimum width=1.2cm, right=0cm of middle, xshift=1.5em, yshift=1.5em] (vrfy) {\small VRFY};
    \node[rect, minimum width=1.2cm, below=0cm of vrfy.south, anchor=north west, yshift=-0.5em] (enc) {\small ENC};
    \node[rect, rounded corners=3mm, minimum width=1.5cm] at ([yshift=1.2cm]$(tag)!0.5!(vrfy)$) (key) {Key};
    \node[rect, minimum width=4.8cm] at ($(middle)!0.5!(Z)$) (A) {\Large$\mathcal{A}$};

\draw[->] (Z.west) to[out=135, in=225] node[midway, left]{\footnotesize$m$} (alice.west);
    \draw[->] ([xshift=0.5em]alice.south west) -- node[midway, left]{\footnotesize$0,1$} ([xshift=0.5em]Z.north west);
    
    \draw[->] ([yshift=0.5em]alice.east) -- node[midway, above]{\footnotesize$m$} (tag.west);
    \draw[<-] ([yshift=-0.5em]alice.east) -- node[midway, below]{\footnotesize$0,1$} ([yshift=-0.5em]alice.east -| comp.west);

    \draw[->] (key) -- node[midway, above left]{\footnotesize$k$} (tag.north);
    \draw[->] (key) -- node[midway, right]{\footnotesize$k$} (comp.north);
    \draw[->] (key) -- node[midway, above right]{\footnotesize$k$} (vrfy.north);
    \draw[->] (key) to[out=270, in=180] node[midway, left]{\footnotesize$k$} (enc.west);

    \draw[->] ([xshift=-1em]tag.south) -- node[near end, left]{\footnotesize$(m,t)$} ([xshift=-1em]tag.south |- A.north);
    \draw[<-] ([xshift=-1em]comp.south) -- node[midway, right]{\footnotesize$\bot,k'$} ([xshift=-1em]comp.south |- A.north);

    \draw[->] ([xshift=-1em]tag.south |- A.south) -- node[midway, left]{\footnotesize$(m,t)$} ([xshift=-1em]tag.south |- Z.north);
    \draw[<-] ([xshift=-1em]comp.south |- A.south) -- node[midway, right]{\footnotesize$\bot,k'$} ([xshift=-1em]comp.south |- Z.north);

    \draw[<-] ([xshift=-1em]vrfy.south) -- node[near end, left]{\footnotesize$(m',t')$} ([xshift=-1em]vrfy.south |- A.north);
    \draw[->] ([xshift=-1em]enc.south) -- node[midway, right]{\footnotesize$\bot,k$} ([xshift=-1em]enc.south |- A.north);

    \draw[<-] ([xshift=-1em]vrfy.south |- A.south) -- node[midway, left]{\footnotesize$(m',t')$} ([xshift=-1em]vrfy.south |- Z.north);
    \draw[->] ([xshift=-1em]enc.south |- A.south) -- node[midway, right]{\footnotesize$\bot,k$} ([xshift=-1em]enc.south |- Z.north);

    \draw[<-] ([yshift=0.5em]bob.west) -- node[midway, above]{\footnotesize$m', \bot$} (vrfy.east);
    \draw[->] ([yshift=-0.5em]bob.west) -- node[midway, below]{\footnotesize$0,1$} ([yshift=-0.5em]bob.west -| enc.east);

    \draw[->] (bob.east) to[out=315, in=45] node[midway, right]{\footnotesize$m',\bot$} (Z.east);

\node[draw,dotted,fit=(key) (tag) (comp) (vrfy) (enc)] {};
\end{tikzpicture} 	\caption{Representation of the real functionality.}\label{real_repr}
\end{figure}
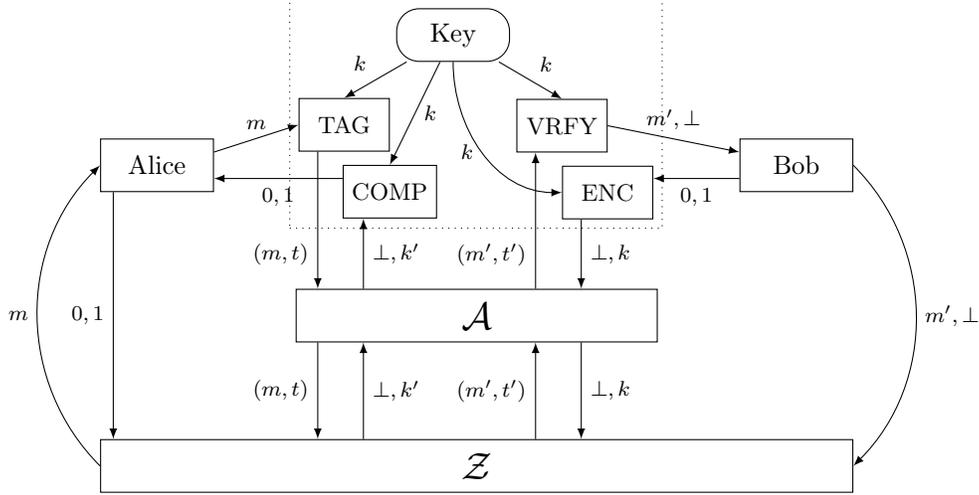

\subsubsection{Real Functionality}\label{subsec:realfunct}
In the real case, depicted in \cref{real_repr}, the functionality is an extension of the one described in Ref.~\cite{AL14}. In particular, we assume a key $k\in\calK$ is shared and used by the algorithm which computes the tag:
\begin{align*}
	\TAG:\calM\times \calK & \to\calM\times \calT \\
	(m,k)                  & \mapsto (m,h_k(m))
\end{align*}
and by the verification algorithm:
\begin{align*}
	\VRFY: \calM \times \calT & \times \calK \to\calM\cup \{\perp\} \\
	(m',t', k)                & \mapsto
	\begin{cases}
		m'    & \text{if } t'=h_k(m'), \\
		\perp & \text{otherwise}.
	\end{cases}
\end{align*}
Moreover, the key is also used in the $\ENC$ and $\COMP$ algorithms.
$\ENC$ encodes the response value $0$ with $\perp$ and $1$ with $k$ while $\COMP$ compares its input $r'\in \{k,\perp\}$ with $k$ and returns $1$ if equality holds or $0$ otherwise. 

\subsubsection{Ideal Functionality}\label{subsec:idealfunct}
In the ideal case, depicted in \cref{ideal_repr}, the sent message $m$ can be modified by an attacker, but any tampering is always signaled to Bob with the symbol $\perp$. The response $r\in\{0,1\}$ (representing communication failure or success, respectively) sent to Alice can again be modified before reaching the other party. Anyway, the ideal functionality $\calF$ returns $0$ to Alice both when Bob’s original input is $0$ and when the message has been tampered, i.e., only $1$ responses are authenticated.

\begin{figure}[tb]
	\centering
	\begin{tikzpicture}[rect/.style={draw, minimum height=0.7cm}]
    \node[rect, minimum width=10cm] at (0,0) (Z) {\Large$\mathcal{Z}$};
    \node[rect, minimum width=1.5cm, above=4cm of Z.west, anchor=west] (alice) {Alice};
    \node[rect, minimum width=1.5cm, above=4cm of Z.east, anchor=east] (bob) {Bob};
    \node[rect, minimum width=4cm] at ($(alice)!0.5!(bob)$) (F) {\Large$\mathcal{F}$};
    \node[rect, minimum width=4cm] at ($(F)!0.5!(Z)$) (S) {\Large$\mathcal{S}$};
    \node[rect, rounded corners=3mm, minimum width=1.5cm, left=1cm of S] (key) {Key};

\draw[->] (Z.west) to[out=135, in=225] node[midway, left]{\small$m$} (alice.west);
    \draw[->] ([xshift=0.5em]alice.south west) -- node[midway, left]{\small$0,1$} ([xshift=0.5em]Z.north west);
    
    \draw[->] ([yshift=0.5em]alice.east) -- node[midway, above]{\small$m$} ([yshift=0.5em]F.west);
    \draw[<-] ([yshift=-0.5em]alice.east) -- node[midway, below]{\small$0,1$} ([yshift=-0.5em]F.west);

    \draw[->] (key) -- node[midway, above]{\small$k$} (S);

    \draw[->] ([xshift=1em]F.south west) -- node[midway, left]{\small$m$} ([xshift=1em]S.north west);
    \draw[<-] ([xshift=2em]F.south west) -- node[midway, right]{\small$0,1$} ([xshift=2em]S.north west);

    \draw[->] ([xshift=1em]S.south west) -- node[midway, left]{\small$(m,t)$} ([xshift=1em]S.south west |- Z.north);
    \draw[<-] ([xshift=2em]S.south west) -- node[midway, right]{\small$\bot,k'$} ([xshift=2em]S.south west |- Z.north);

    \draw[<-] ([xshift=-2em]F.south east) -- node[midway, left]{\small$m'$} ([xshift=-2em]S.north east);
    \draw[->] ([xshift=-1em]F.south east) -- node[midway, right]{\small$0,1$} ([xshift=-1em]S.north east);

    \draw[<-] ([xshift=-2em]S.south east) -- node[midway, left]{\small$(m',t')$} ([xshift=-2em]S.south east |- Z.north);
    \draw[->] ([xshift=-1em]S.south east) -- node[midway, right]{\small$\bot,k$} ([xshift=-1em]S.south east |- Z.north);

    \draw[->] ([yshift=0.5em]F.east) -- node[midway, above]{\small$m',\bot$} ([yshift=0.5em]bob.west);
    \draw[<-] ([yshift=-0.5em]F.east) -- node[midway, below]{\small$0,1$} ([yshift=-0.5em]bob.west);

    \draw[->] (bob.east) to[out=315, in=45] node[midway, right]{\small$m',\bot$} (Z.east);
\end{tikzpicture} 	\caption{Representation of the ideal functionality.}\label{ideal_repr}
\end{figure}
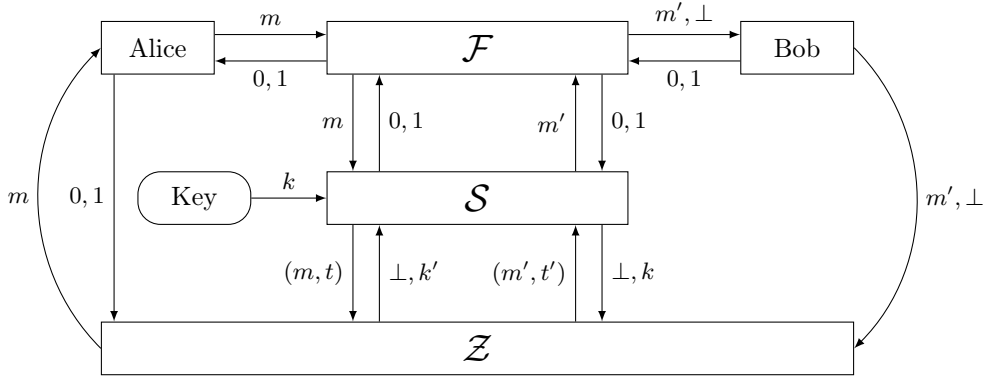

\subsection{Distinguisher}\label{subsec:disting}
In order to test the indistinguishability of the two functionalities, we consider the two systems involving the distinguisher. The distinguisher, denoted by $\calZ$, controls the message $m$ to be authenticated. We define $X$ to be the random variable whose realization is $m \in \calM$, with $\calM$ the set of all possible messages.

In the real setup, the pair $(m,t)\in \calM \times \calT$ is sent through the channel, which is identified with an \emph{adversary}, denoted by $\calA$. We can assume $\calA$ to be completely controlled by the distinguisher, hence $\calA$ redirects messages to and from the distinguisher. In particular, given $\calT$ the set of the possible tags, we denote as $Y$ the random variable whose realization is the pair $(m,t) \in \calM \times \calT$, and with $Y'$ the one related to $(m',t')  \in \calM \times \calT$ whose distribution is chosen by the distinguisher.

In the ideal case we use the same symbols $X$, $Y$ and $Y'$ for the random variables whose realizations are $m\in \calM$, $(m,t)\in \calM \times \calT$ and $(m',t')\in \calM \times \calT$, respectively. Notice that in the ideal setup the functionality is equipped with a simulator $\calS$ which avoids trivial distinguishment between real and ideal case, adding tags and properly encoding the response message.

Assuming that the simulator $\calS$ and the real functionality $\calA$ use the same key distribution, it follows that the random variables $X$, $Y$ and $Y'$ also have the same distributions in the two cases.

In order to describe the interactions of the distinguisher, we also introduce the random variables $\tilde{X}$ representing $x'\in \calM \cup \{\perp\}$ that the environment gets from Bob in the real case, and $\hat{X}$ representing $x'\in \calM \cup \{\perp\}$ that the environment gets from Bob in the ideal one.

The description provided lists random variables related to the authenticated transmission of $m$ and can be found with the same symbols also in Ref.~\cite{AL14}.
We now add the additional random variables needed in our extension of the scheme. We assume that the response $r \in \calK \cup \{\perp\}$ passes through $\calA$ ($\calS$) in the real (ideal) case and corresponds to the random variable $\tilde{Z}$ ($\hat{Z}$). Moreover, we denote by $\tilde{Z}'$ ($\hat{Z}'$) the eventually tampered response $r' \in \calK \cup \{\perp\}$ that the environment returns to $\calA$ ($\calS$).
Finally, we denote with $\tilde{F}$ and $\hat{F}$ the random variables in the real and ideal case, respectively, associated to Alice's final result $f \in \{0,1\}$.

\subsection{Security Proof}
Since the introduced random variables describe all of the distinguisher interactions with the system, we are now able to state our result in terms of trace distance of the joint probability distributions.
\begin{theorem}[Authentication-with-response]\label{thrm:main}
	With the notation and assumptions of \cref{subsec:disting} in place, no distinguisher $\calZ$ can distinguish between the two cases:
	\begin{enumerate}
		\item it is interacting with $\calA$ and participants running the authentication-with-response scheme based on the set of $\varepsilon$-ASU\textsubscript{2} functions $\calH=\{ h_{k}:\calM \to \calT\}_{k\in\calK}$, and using $\varepsilon'$-perfect keys,
		\item it is interacting with $\calS$ and participants running $\calF$,
	\end{enumerate}
	except with probability $\frac{1}{2}\left(1+\frac{|\calT|}{|\calK|} + \varepsilon + \varepsilon'\right)$.\\
	Equivalently, the following upper bound for the trace distance holds:
	\begin{equation*}\delta(P_{XYY'\tilde{X}\tilde{Z}\tilde{Z'}\tilde{F}}, P_{XYY'\hat{X}\hat{Z}\hat{Z'}\hat{F}}) \le \frac{|\calT|}{|\calK|} + \varepsilon + \varepsilon'.
	\end{equation*}
	where $P_{XYY'\tilde{X}\tilde{Z}\tilde{Z'}\tilde{F}}$ and $P_{XYY'\hat{X}\hat{Z}\hat{Z'}\hat{F}}$ are joint probability distributions.
\end{theorem}

\begin{proof}
	Denote by $\tilde{P} = P_{XYY'\tilde{X}\tilde{Z}\tilde{Z'}\tilde{F}}$ and $\hat{P} = P_{XYY'\hat{X}\hat{Z}\hat{Z'}\hat{F}}$ the joint probabilities. The trace distance is given by
	\begin{equation}
		\label{eq:trace_proof}
		\delta(\tilde{P}, \hat{P})
		= \frac{1}{2} \sum_{m, y, y', x', z, z', f} \left|\tilde{P}(m, y, y', x', z, z', f) - \hat{P}(m, y, y', x', z, z', f)\right|.
	\end{equation}
	The trace distance can be divided in two terms, $\delta(\tilde{P}, \hat{P}) = \delta_{n}(\tilde{P}, \hat{P}) + \delta_{t}(\tilde{P}, \hat{P})$, corresponding to two distinct strategies that the distinguisher can adopt, i.e., attempting or not to tamper the pair $(m,t)\in \calM \times \calT$.

	\paragraph{No tampering.}
	We first analyze the case in which there is no tampering attempt, i.e., term $\delta_{n}(\tilde{P}, \hat{P})$, given by:
	\begin{align*}
		\delta_{n}(\tilde{P}, \hat{P})  = \; & \frac{1}{2} \sum_{m, y, y'=y, x', z, z', f} \left|\tilde{P}(m, y, y', x', z, z', f) - \hat{P}(m, y, y', x', z, z', f)\right| \\
		= \;                                 & \frac{1}{2} \sum_{m, y, x', z, z', f} \left|\tilde{P}(m, y, y, x', z, z', f) - \hat{P}(m, y, y, x', z, z', f)\right|.
	\end{align*}
	Because of the system description, i.e., due to the fact that $y'=y$ implies $\tilde{X} = \hat{X} = X$ and that $\tilde{Z} = \hat{Z}$, the term $\delta_{n}(\tilde{P}, \hat{P})$ can also be written as:

	\begin{align*}
\frac{1}{2} \sum_{m, y, z, z', f} P_{XYY'\tilde{X}\tilde{Z}}(m, y, y,m, z)
		\left| P_{\tilde{Z}'\tilde{F}\mid XYY'\tilde{X}\tilde{Z}}(z', f \mid m, y, y,m, z)- P_{\hat{Z}'\hat{F}\mid XYY'\hat{X}\hat{Z}}(z', f \mid m, y, y,m, z)\right|.
	\end{align*}
	Moreover, we have the following equality:
	\[
		P_{\tilde{Z}'\mid XYY'\tilde{X}\tilde{Z}}(z'\mid  m, y, y,m,z)  = P_{\hat{Z}'\mid XYY'\hat{X}\hat{Z}}(z'\mid  m, y, y,m,z).
	\]
	Indeed, since everything goes in the same way before selecting \(z'\), there is no reason why the distinguisher would choose a different \(z'\) depending on whether it is interacting with the ideal or real functionality. In other words, the condition $Y = Y' = y$ implies that also $\tilde{Z}' = \hat{Z}'$. Hence, $\delta_{n}(\tilde{P}, \hat{P})$ further simplifies to:

	\[
		\frac{1}{2} \sum_{m, y, z, z', f} P_{XYY'\tilde{X}\tilde{Z}\tilde{Z}'}(m, y, y,m, z,z')
		\left| P_{\tilde{F}\mid XYY'\tilde{X}\tilde{Z}\tilde{Z'}}(f \mid m, y, y,m, z,z')- P_{\hat{F}\mid XYY'\hat{X}\hat{Z}\hat{Z}'}(f \mid m, y, y,m, z,z')\right|.
	\]
	Finally, since the computation of the two functionalities is equal in this case, it follows that:

	\[
		P_{\tilde{F}'\mid XYY'\tilde{X}\tilde{Z}\tilde{Z}'}(f\mid  m, y, y,m, z, z') = P_{\hat{F}'\mid XYY'\hat{X}\hat{Z}\hat{Z}'}(f\mid  m, y, y,m, z, z'),
	\]
	concluding that the no tampering component $\delta_{n}(\tilde{P}, \hat{P})$ vanishes.

	\paragraph{Tampering.}
	The trace distance of \cref{eq:trace_proof} therefore reduces to $\delta(\tilde{P}, \hat{P}) = \delta_{t}(\tilde{P}, \hat{P})$. The term $\delta_{t}(\tilde{P}, \hat{P})$ is given by:
	\begin{equation*}\delta_{t}(\tilde{P}, \hat{P})=\frac{1}{2} \sum_{m, y, y'\not=y, x', z, z', f} \left|\tilde{P}(m, y, y', x', z, z', f) - \hat{P}(m, y, y', x', z, z', f)\right|,
	\end{equation*}
	which represents the distinguisher's attempt to determine in which system it is tampering the first sent message.
	Given $y=(m,t)\in \calM \times\calT$, define the following subset of $\calK$:
	\[
		A_y = \{k\in \calK \mid h_k(m)=t\}\subseteq \calK.
	\]
	and denote by
$A^{\mathsf{c}}_{y}\subseteq \calK$ its complementary
sets in the key space $\calK$, hence
	\begin{align*}
		\tilde{P}(m, y, y', x', z, z', f) = \tilde{P}_{A_{y'}}(m, y, y', x', z, z', f)+\tilde{P}_{A^{\mathsf{c}}_{y'}}(m, y, y', x', z, z', f)
	\end{align*}
	with
	\begin{align*}
		\tilde{P}_{A_{y'}}(m, y, y', x', z, z', f)              & = \Pr[X=m, Y=y, Y'=y', \tilde{X} = x', \tilde{Z}=z, \tilde{Z}'=z', \tilde{F}=f, K\in A_{y'}]              \\
		\tilde{P}_{A^{\mathsf{c}}_{y'}}(m, y, y', x', z, z', f) & =\Pr[X=m, Y=y, Y'=y', \tilde{X} = x', \tilde{Z}=z, \tilde{Z}'=z', \tilde{F}=f, K\in A^{\mathsf{c}}_{y'}].
	\end{align*}
	Analogously, the same can be done  for $\hat{P}_A$ and $\hat{P}_{A^{\mathsf{c}}}$:
	\[
		\hat{P}(m, y, y', x', z, z', f) = \hat{P}_{A_{y'}}(m, y, y', x', z, z', f)+\hat{P}_{A^{\mathsf{c}}_{y'}}(m, y, y', x', z, z', f).
	\]
	Then, we use the triangular inequality:
	\begin{equation}\label{eq: tr in}
		\begin{split}
\delta_{t}(\tilde{P}, \hat{P})
			\le \  & \frac{1}{2} \left[\sum_{m, y, y'\not=y, x', z, z', f}\left|\tilde{P}_{A_{y'}}(m, y, y', x', z, z', f) - \hat{P}_{A_{y'}}(m, y, y', x', z, z', f)\right| \right] +                                \\
			       & + \frac{1}{2}   \left[\sum_{m, y, y'\not=y, x', z, z', f}  \left|\tilde{P}_{A^{\mathsf{c}}_{y'}}(m, y, y', x', z, z', f) - \hat{P}_{A^{\mathsf{c}}_{y'}}(m, y, y', x', z, z', f)\right| \right].
		\end{split}
	\end{equation}
	Consider the first term of \cref{eq: tr in} (first line). Due to the condition given by the set \(A_{y'}\), where \(y'= (m',t')\), the following hold:
	\begin{align*}
		\hat{P}_{A_{y'}}(m, y, y', x', z, z', f)   & = 0 \text{ if } x'\neq \perp; \\
		\tilde{P}_{A_{y'}}(m, y, y', x', z, z', f) & = 0 \text{ if } x'\neq m'.
	\end{align*}
	Namely, when $y'\neq y$ the ideal functionality recognizes the tampering and considers the authentication failed. On the other hand, the real functionality sets $\tilde{X}=m'$ because of \(h_K(m')=t'\). Moreover, specifying the set \(A_{y'}\) becomes redundant if \(\tilde{X} = m'\). Therefore, the first term of \cref{eq: tr in} simplifies to
	\begin{align}\label{eq: tr in 1}
		\sum_{m,y,y'\neq y, z, z',f} \tilde{P}(m, y, y', m', z, z', f) & = \sum_{m,y,y'\neq y} \tilde{P}(m, y, y', m') \nonumber                             \\
		                                                               & =\; \sum_m P_X(m)\sum_{t, y'\neq y}P_{Y'\mid Y}(y'\mid y)\Pr[h_K(m)=t, h_K(m')=t'],
	\end{align}
where the first equality comes from the total probability law and the second is proved in~\cite[Theorem 5]{AL14}\footnote{Observe that the resulting term is exactly the same.}.\\
	The second term of \cref{eq: tr in} (second line) corresponds to the case $(m',t') \not=(m, t)$ and both, the $\varepsilon$-ASU\textsubscript{2} authentication scheme and the ideal functionality recognize the tampering. \\
	In the real functionality, the probability of the set
	\[
		\{X=m, Y=y, Y'=y', \tilde{X'} = x', \tilde{Z} = z, \tilde{Z}' = z', \tilde{F} =f, K\in A^{\mathsf{c}}_{y'}\}
	\]
	is the same as the set
	\[
		\{X=m, h_K(m)=t, Y'=y', \tilde{X'} = x', \tilde{Z} = z, \tilde{Z}' = z', \tilde{F} =f, h_K(m')\neq t'\}.
	\]
	Furthermore, the condition  $h_K(m') \not= t' $ implies that $\tilde{X'} = \perp$ and $\tilde{Z} = \perp $, then a (eventually) nonzero probability set can be written as
	\[
		\{X=m, h_K(m)=t, Y'=y', \tilde{X'} = \perp, \tilde{Z} = \perp, \tilde{Z}' = z', \tilde{F} =f, h_K(m')\not=t'\}
	\]
	and the probability is the same as the set
	\begin{equation}\label{set1}
		\{X=m, Y'=y', h_K(m)=t, h_K(m')\not=t', \tilde{Z}' = z',\ \tilde{F} =f\}.
	\end{equation}
	The same happens for the ideal functionality events involved in the sum; the set to be considered is
	\begin{equation}\label{set2}
		\{X=m, Y'=y', h_K(m)=t, h_K(m')\not=t', \hat{Z}' = z',\ \hat{F} =f\}.
	\end{equation}
	In particular, the term in the second line of \cref{eq: tr in} can be simplified, in terms of probabilities of the sets~(\ref{set1}) and (\ref{set2}):
	\begin{equation}
		\begin{split}\label{eq:intermediate1}
			 & \frac{1}{2} \sum_{m, t, y'\neq y, x', z, z', f}  \left|\tilde{P}_{A^{\mathsf{c}}}(m, y, y', x', z, z', f) - \hat{P}_{A^{\mathsf{c}}}(m, y, y', x', z, z', f)\right| \\
			 & =\; \frac{1}{2} \sum_{m, t,y'\neq y, z', f}  \Big|\Pr[X=m, Y' = y', h_K(m)=t, h_K(m')\not=t',\tilde{Z}' = z', \tilde{F} =f] +                                       \\
			 & \hspace{2cm}  - \Pr[X=m, Y' = y', h_K(m)=t, h_K(m')\not=t',\hat{Z}' = z', \hat{F} =f]\Big|.
		\end{split}
	\end{equation}
	Defining the set $E(m,t,y'):=\{X=m, Y' = y', h_K(m)=t, h_K(m')\not=t'\}$, \cref{eq:intermediate1} can be rewritten as
	\begin{align}\label{eq:intermediate2}
		 & \frac{1}{2} \sum_{m, t, y'\neq y,z', f}\Pr[E(m,t,y')] \left|
		\Pr[\tilde{Z}'=z',\tilde{F}=f\mid E(m,t,y')]-\Pr[\hat{Z}'=z',\hat{F}=f\mid E(m,t,y')]
		\right|    \nonumber                                                                                    \\
		 & = \frac{1}{2} \sum_{m, t, y'\neq y, z', f}\Pr[E(m,t,y')] \Big|
		\Pr[\tilde{F}=f\mid \tilde{Z}'=z',E(m,t,y')]\Pr[\tilde{Z}'=z'\mid E(m,t,y')]+                           \\
		 & \hspace{2cm}  -\Pr[\hat{F}=f\mid \hat{Z}'=z',E(m,t,y')]\Pr[\hat{Z}'=z'\mid E(m,t,y')]\Big|.\nonumber
	\end{align}

	\begin{remark}
		Given $E(m,t,y')$ there is no reason why the real and the ideal case differ. We have therefore:
		\[
			\Pr[\hat{Z}'=z'\mid E(m,t,y')] = \Pr[\tilde{Z}'=z'\mid E(m,t,y')].
		\]
		Hence, \cref{eq:intermediate2} simplifies to
		\begin{equation}\label{eq:intermediate3}
			\frac{1}{2} \sum_{m, t,y'\neq y,z', f}\Pr[\tilde{Z}'=z', E(m,t,y')] \left| \Pr[\tilde{F}=f|\tilde{Z}'=z',E(m,t,y')]-\Pr[\hat{F}=f|\hat{Z}'=z',E(m,t,y')]\right|.
		\end{equation}
	\end{remark}
	\begin{remark}
		Since in the ideal case the tampering is always detected regardless of the value of $z'$, for any $z'$, we have
		\begin{align*}
			\Pr[\hat{F}=1\mid \hat{Z}'=z', E(m,t,y')] & = 0, \\
			\Pr[\hat{F}=0\mid \hat{Z}'=z', E(m,t,y')] & = 1.
		\end{align*}
	\end{remark}
	\noindent Therefore, expanding \cref{eq:intermediate3} over all possible values of $f$, we get
	\begin{align}\label{eq:intermediate4}
		\frac{1}{2} & \sum_{m, t,y'\neq y, z'}\Pr[\tilde{Z}'=z', E(m,t,y')] \left[1-\Pr[\tilde{F}=0 \mid \tilde{Z}'=z',E(m,t,y')]+\Pr[\tilde{F}=1 \mid \tilde{Z}'=z',E(m,t,y')]\right] \nonumber           \\
		= \;        & \sum_{m, t,y'\neq y,z'}\Pr[\tilde{Z}' = z', E(m,t,y')]\Pr[\tilde{F}=1 \mid \tilde{Z}'=z', E(m,t,y')]                                                                       \nonumber \\
		=\;         & \sum_{m, t,y'\neq y, z'}\Pr[\tilde{F}=1,\tilde{Z}'=z',E(m,t,y')]                                                                                                          \nonumber  \\
		=\;         & \sum_{m,t,y'\neq y}\Pr[\tilde{F}=1, E(m,t,y')],
	\end{align}
	since $\Pr[\tilde{F}=0 \mid \tilde{Z}'=z',E(m,t,y')] = 1- \Pr[\tilde{F}=1 \mid \tilde{Z}'=z',E(m,t,y')]$ and the total probability law is applied to get rid of $\tilde{Z}'$. Then, observing that the knowledge of \(X=m\) and \(Y'=y'\) do not affect \(\tilde{F}\) and \(K\), \cref{eq:intermediate4} becomes
	\begin{align*}
		 & \sum_{m,t,y'\neq y}\Pr[\tilde{F}=1, E(m,t,y')] = \sum_{m,t,y'\neq y}\Pr[\tilde{F}=1, h_K(m)=t, h_K(m') \neq t', Y'=y',X=m]         \\
		 & =\;   \sum_{m,t,y'\neq y}\Pr[X=m]\Pr[h_K(m)=t]\Pr[Y'=y'\mid X=m, h_K(m)=t]\Pr[\tilde{F}=1, h_K(m') \neq t'| X=m, Y'= y', h_K(m)=t] \\
		 & =\;   \sum_{m,t,y'\neq y}\Pr[X=m]\Pr[Y'=y'\mid Y = y]\Pr[h_K(m)=t]\Pr[\tilde{F}=1, h_K(m') \neq t'| h_K(m)=t]                      \\
		 & =\;  \sum_{m,t,y'\neq y}\Pr[X=m]\Pr[Y'=y'\mid Y = y]\Pr[\tilde{F}=1, h_K(m)=t, h_K(m') \neq t'].
	\end{align*}
	Moreover, considering that to realize \(\tilde{F}=1\), a distinguisher has to guess a key, \cref{eq:intermediate4} can be bounded as follows:
	\begin{align}
		\sum_{m,t, y'\neq y}\Pr[\tilde{F}=1, E(m,t,y')] & \leq \sum_{m}P_{X}(m)\sum_{t,y' \neq y}P_{Y'|Y}(y'|y)\max_{k\in A_y\cap A^c_{y'}}\Pr[K=k] \nonumber \\
		                                                & = \sum_{m}P_{X}(m)\sum_{t,y' \neq y}P_{Y'|Y}(y'|y)\Pr[K=k_{y,y'}]. \label{eq: tr in 2}
	\end{align}
	The second equality is given by the observation that there exists, possibly not unique, a \(k_{y,y'}\) realizing the maximum, i.e., such that
	\[
		\Pr[K=k_{y,y'}] = \max_{k\in A_y\cap A^c_{y'}}\Pr[K=k].
	\]

	\paragraph{Overall bound.} Combining Equations~(\ref{eq: tr in 1}) and (\ref{eq: tr in 2}) we have a bound for the trace distance: \begin{equation}\label{eq:int_final}
		\delta(\tilde{P}, \hat{P}) = \delta_{t}(\tilde{P}, \hat{P})\le\sum_{m}P_{X}(m)\sum_{t,y' \neq y}P_{Y'|Y}(y'|y)\Big(\Pr[h_K(m)=t,h_K(m')=t']+\Pr[K=k_{y,y'}]\Big).
	\end{equation}
	The term $P_{Y'|Y}((m',t')|(m,t))$ corresponds to the attack strategy. Assume that this is deterministic, meaning that $(m',t')$ is a function of $(m,t)$, then:

	\begin{align}\label{eq:inter0}
		     & \sum_{t,y' \neq y}P_{Y'|Y}(y'|y)\Big(\Pr[h_K(m)=t,h_K(m')=t']+\Pr[K=k_{y,y'}]\Big) \nonumber \\
		= \; & \sum_{t}\Pr[h_K(m)=t,h_K(m'(m,t))=t'(m,t)]+\Pr[K=k_{y,y'(m,t)}].
	\end{align}
	By construction, \(t_1\neq t_2\) implies that the sets $\{h_K(m)=t_1\}$ and $\{h_K(m)=t_2\}$ are disjoint, hence \cref{eq:inter0} can be further elaborated to
	\begin{equation*}\label{eq:inter1}
		\Pr\left[\bigsqcup_t\{h_K(m)=t, h_K(m'(m,t))=t'(m,t)\}\right]  +  \Pr\left[\bigsqcup_t\{K=k_{y,y'(m,t)}\}\right].
	\end{equation*}
	Moreover, the different condition on $h_K(m'(m,t))$ and $t'(m,t)$ between the two terms of the sum provides another events disjunction, further simplifying the probability to
	\begin{equation}\label{eq:inter2}
		\Pr\left[\bigsqcup_t \Big(\{h_K(m)=t, h_K(m'(m,t))=t'(m,t)\} \sqcup \{K=k_{y,y'(m,t)}\} \Big)\right].
	\end{equation}
	Notice that this is the probability of an event of $\mathcal{P}(\calK)$.
	In order to bound the probability of \cref{eq:inter2}, we give a bound on the number of events.
	By \cref{def:epsASU} of $\varepsilon$-ASU\textsubscript{2} functions, we have
\[
		\Big|\{h_K(m)=t, h_K(m'(m,t))=t'(m,t)\}\Big| \le \varepsilon \frac{|\calK|}{|\calT|}
	\]
while $|\{K=k_{y,y'(m,t)}\}| = 1$. Combining these two cardinalities, we get the upper bound:
	\begin{equation}\label{eq:cards}
		\left|\bigsqcup_t \Big(\{h_K(m)=t, h_K(m'(m,t))=t'(m,t)\} \sqcup \{K=k_{y,y'(m,t)}\} \Big)\right|\le |\calT|\left(\varepsilon\frac{|\calK|}{|\calT|} + 1\right) = |\calT| + \varepsilon |\calK|.
	\end{equation}
	Using the cardinality bound (\ref{eq:cards}), and recalling the assumption of $\varepsilon'$-perfectness of the keys and the result in~\cite[Lemma 1]{AL14}, we further bound \cref{eq:int_final} to

	\begin{equation*}
		\begin{split}
			\delta(\tilde{P}, \hat{P}) & \leq \; \sum_{m}\Pr[X=m]\left(\frac{|\calT| + \varepsilon |\calK|}{|\calK|} + \varepsilon'\right)       \\
			                           & = \; \left(\frac{|\calT|}{|\calK|} + \varepsilon + \varepsilon'\right) \underbrace{\sum_m\Pr[X=m]}_{=1} \\
			                           & = \; \frac{|\calT|}{|\calK|} + \varepsilon + \varepsilon'.
		\end{split}
	\end{equation*}
	Observe that the resulting bound does not depend on the particular attack strategy selected by the deterministic approach.
	\paragraph{Probabilistic attacks. }
	To extend the result to a probabilistic attack, we strictly follow the proof of~\cite[Theorem 5]{AL14}. Introduce an auxiliary probability space $(\Omega, \mathcal{B},\mu)$ for the random variable $Y'=(X',T')$, where $\Omega$ is the sample space, $\mathcal{B}$ is the $\sigma$-algebra of events and $\mu$ is the probability measure. Using the indicator function $\chi$, we have
	\[
		P_{Y'|Y}((m',t')|(m,t))=\int_{\Omega}\chi_{\{\omega\in\Omega\mid Y'(m,t,\omega)=(m',t')\}}(\omega)\,d\mu.
	\]
	Then
	\begin{align*}
		\delta(\tilde{P}, \hat{P}) \leq \; & \sum_{m}P_{X}(m)\sum_{t,y' \neq y}P_{Y'|Y}(y'|y)\left(\Pr[h_K(m)=t,h_K(m')=t']+\Pr[K=k_{y,y'}]\right)                                                                                          \\
		=\;                                & \sum_{m}P_{X}(m)\sum_{t,y' \neq y}\int_{\Omega}\chi_{\{\omega\in\Omega\mid Y'(m,t,\omega)=(m',t')\}}(\omega)\,d\mu\left(\Pr[h_K(m)=t,h_K(m')=t']+\Pr[K=k_{y,y'}]\right)                        \\
		=\;                                & \int_{\Omega}\sum_{m}P_{X}(m)\sum_t\underbrace{\sum_{y'\neq y}\chi_{\{\omega\in\Omega\mid Y'(m,t,\omega)=(m',t')\}}(\omega)}_{(*)}\left(\Pr[h_K(m)=t,h_K(m')=t']+\Pr[K=k_{y,y'}]\right)\,d\mu.
	\end{align*}
	For each fixed \(m\), \(t\) and \(\omega\), the unique non-zero term of the sum $(*)$ occurs with \((m',t')=Y'(m,t,\omega)\), resulting in a deterministic attack. Hence, the equation simplifies to
	\[
		\int_{\Omega}\sum_{m}P_{X}(m)\sum_t\Pr[h_K(m)=t,h_K(X'(m,t,\omega))=T'(m,t,\omega)]+\Pr[K=k_{y,Y'(m,t,\omega)}]\,d\mu.
	\]
	The proof is now completed, applying the above approach for the deterministic strategy
	\begin{align*}
		\delta(\tilde{P}, \hat{P}) \leq \; & \int_{\Omega}\sum_{m}P_{X}(m)\sum_t\Pr[h_K(m)=t,h_K(X'(m,t,\omega))=T'(m,t,\omega)]+\Pr[K=k_{y,Y'(m,t,\omega)}]\,d\mu           \\
		\leq                               & \int_{\Omega}\frac{|\calT|}{|\calK|} + \varepsilon + \varepsilon'\,d\mu = \frac{|\calT|}{|\calK|} + \varepsilon + \varepsilon'.
	\end{align*}

\end{proof}
\cref{thrm:main} directly implies the $\big(\frac{|\calT|}{|\calK|} + \varepsilon + \varepsilon'\big)$-security of the {authentication-with-response} scheme.
In order to evaluate the found bound, we compare it with the one provided in Ref.~\cite{AL14}: the additional response phase (steps~\ref{resp} and~\ref{ver-res} of our scheme) adds a $\frac{|\calT|}{|\calK|}$ contribution.
Notice that, in the case of polynomial hashing~\cite{BO93}, the key space is given by $\calK = \calT \times \calT$, hence the additional contribution reduces to $\frac{|\calT|}{|\calK|}=\frac{1}{|\calT|} < \varepsilon$.

\section{Conflict of interest statement}
The authors declare no conflicts of interest.

\section{Acknowledgments}
The Authors acknowledge financial support from EQUO (European QUantum ecOsystems) project, funded by the European Commission in the Digital Europe Programme [GA No. 101091561]. DB acknowledges support from the European Union ERC StG, QOMUNE, 101077917. AZ acknowledges support from the European Union - NextGeneration EU, "Integrated infrastructure initiative in Photonic and Quantum Sciences" - I-PHOQS [IR0000016, ID D2B8D520, CUP B53C22001750006]. 

\bibliographystyle{ieeetr}
\bibliography{references}

\appendix
\section{Theoretical Advantage in Multiple QKD Rounds}\label{app:mult:round}
We investigate security bounds to the UC security of the QKD protocol after multiple rounds. In particular, we compare bounds provided by the straightforward mutual authentication with those of the authentication-with-response variant.

\subsection{Security Bound to the Straightforward Approach}

In the literature, it is usually discussed that due to the Universal Composability theorem, combining an $\varepsilon_1$-secure quantum key exchange with an $\varepsilon_2$-secure authentication step provides an $(\varepsilon_1+\varepsilon_2)$-secure QKD protocol.
Actually, this reasoning is missing a crucial point. Regardless of the above-described {authentication-with-response} variant, to obtain the needed mutual authentication between parties, a single authentication step is not enough. Indeed, at least a couple of them is required, one from one party to the other and vice versa, as depicted in \cref{fig:multi-round-straightforward}. We refer to it as the \emph{straightforward approach}.

Therefore, the previous statement should be that combining an $\varepsilon_1$-secure quantum key exchange with two $\varepsilon_2$-secure authentication steps provides an $(\varepsilon_1+2\varepsilon_2)$-secure QKD protocol~\cite{molotkov2024many}. This factor $2$ becomes critical in evaluating the theoretical bound to the QKD protocol security as long as the QKD rounds go on. This is due to the fact that, while the security parameter $\varepsilon_1$ remains the same along all QKD rounds, $\varepsilon_2$ grows depending on the knowledge of the involved authentication key.

To evaluate the above-mentioned theoretical bound and to avoid inconsistency, it is useful to define as:
\begin{itemize}
	\item $\varepsilon_{2,i}$ the $\varepsilon_2$ corresponding to the $i$-th QKD round;
	\item $\varepsilon'_i$ the perfectness parameter of the authentication key used in the $i$-th QKD round;
	\item $\tilde{\varepsilon}_i$ the security parameter of the overall $i$-th QKD round.
\end{itemize}

Assume that, at the very first QKD round, perfect authentication keys are used, i.e., $\varepsilon'_1$-perfect keys with $\varepsilon'_1=0$, and that, for the whole process, authentication is performed via $\varepsilon$-ASU\textsubscript{2} functions. Notice that, for subsequent rounds, the authentication key comes from the previous one. Following the results of~\cite{AL14}, the following relations hold:
\begin{equation*}
	\varepsilon'_i = \begin{cases}
		0                         & \text{for } i=1  \\
		\tilde{\varepsilon}_{i-1} & \text{for }i > 1
	\end{cases},
	\qquad \varepsilon_{2,i} = \varepsilon + \varepsilon'_i,
	\qquad \tilde{{\varepsilon}}_i = \varepsilon_1 + 2\varepsilon_{2,i}.
\end{equation*}

The general formula regarding the $i$-th round in terms of $\varepsilon_1$ and $\varepsilon$ can be easily proved by induction. In particular, we obtain the following:
\begin{equation*}
	\varepsilon'_i=(2^{i-1}-1)\varepsilon_1+(2^{i}-2)\varepsilon, \qquad
	\varepsilon_{2,i}=(2^{i-1}-1)\varepsilon_1+(2^{i}-1)\varepsilon, \qquad
	\tilde{\varepsilon}_i=(2^{i}-1)\varepsilon_1+(2^{i+1}-2)\varepsilon.
\end{equation*}

We notice that the theoretical bound $\tilde{\varepsilon}_i$, representing the overall security of the $i$-th round, is exponential in $i$. This renders this approach impracticable for a large number of rounds. While this does not prove that performing multiple QKD rounds is insecure, we point out that, to the best of our knowledge, a complete security proof for concatenating QKD rounds with mutual authentication in the straightforward approach is missing.

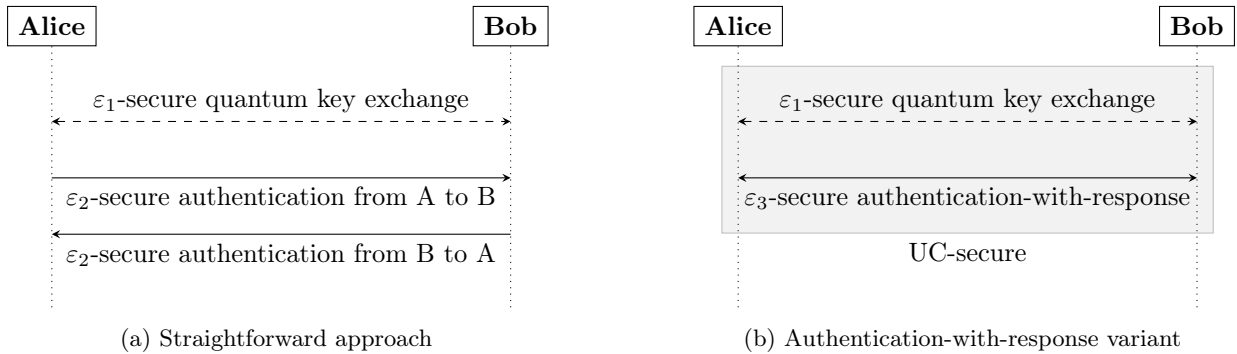
\begin{figure}[tb]
	\centering
	\begin{subfigure}{0.45\textwidth}
		\centering
		\begin{tikzpicture}[
		>=stealth,
		party/.style      = {draw, fill=white, inner sep=4pt, font=\bfseries},
		note/.style       = {inner sep=2pt, align=left},
		annotation/.style = {note, font=\footnotesize, gray},
		msg/.style        = {->},
		phasebox/.style   = {fill=gray!10, draw=gray!50},
		inbox/.style      = {draw, rounded corners, fill=white, inner sep=2pt},
		num/.style        = {circle, draw, fill=white, inner sep=1.5pt, font=\footnotesize\bfseries},
	]

	\node[party] (A) {Alice};
	\node[party, right=5cm of A] (B) {Bob};

	\coordinate[below=1cm of A] (Ar0);
	\coordinate (Br0) at (B |- Ar0);
	\draw[<->, dashed] (Ar0) -- node[above, midway]{$\varepsilon_1$-secure quantum key exchange} (Br0);

	\coordinate[below=0.75cm of Ar0] (Ar1);
	\coordinate (Br1) at (B |- Ar1);
	\draw[->] (Ar1) -- node[below, midway]{$\varepsilon_2$-secure authentication from A to B} (Br1);

	\coordinate[below=0.75cm of Ar1] (Ar2);
	\coordinate (Br2) at (B |- Ar2);
	\draw[<-] (Ar2) -- node[below, midway]{$\varepsilon_2$-secure authentication from B to A} (Br2);

	\coordinate[below=1cm of Ar2] (A-bot);
	\coordinate (B-bot) at (B |- A-bot);
	\draw[dotted] (A.south) -- (A-bot);
	\draw[dotted] (B.south) -- (B-bot);
\end{tikzpicture}
 		\caption{Straightforward approach}\label{fig:multi-round-straightforward}
	\end{subfigure}
	\hfill
	\begin{subfigure}{0.45\textwidth}
		\centering
		\begin{tikzpicture}[
		>=stealth,
		party/.style      = {draw, fill=white, inner sep=4pt, font=\bfseries},
		note/.style       = {inner sep=2pt, align=left},
		annotation/.style = {note, font=\footnotesize, gray},
		msg/.style        = {->},
		phasebox/.style   = {fill=gray!10, draw=gray!50},
		inbox/.style      = {draw, rounded corners, fill=white, inner sep=2pt},
		num/.style        = {circle, draw, fill=white, inner sep=1.5pt, font=\footnotesize\bfseries},
	]

	\node[party] (A) {Alice};
	\node[party, right=5cm of A] (B) {Bob};

	\coordinate[below=1cm of A] (Ar0);
	\coordinate (Br0) at (B |- Ar0);
	\draw[<->, dashed] (Ar0) -- node[above, midway] (lbl-0) {$\varepsilon_1$-secure quantum key exchange} (Br0);

	\coordinate[below=0.75cm of Ar0] (Ar1);
	\coordinate (Br1) at (B |- Ar1);
	\draw[<->] (Ar1) -- node[below, midway] (lbl-1) {$\varepsilon_3$-secure authentication-with-response} (Br1);

	\coordinate[below=1.75cm of Ar1] (A-bot);
	\coordinate (B-bot) at (B |- A-bot);
	\draw[dotted] (A.south) -- (A-bot);
	\draw[dotted] (B.south) -- (B-bot);

	\begin{scope}[on background layer]
		\node[phasebox, inner sep=5pt,
			fit=(Ar0)(Ar1)(Br0)(Br1)(lbl-0)(lbl-1),
			label={below: UC-secure}] {};
	\end{scope}
\end{tikzpicture}
 		\caption{Authentication-with-response variant}\label{fig:multi-round-auth-with-resp}
	\end{subfigure}
	\caption{Mutually authenticated QKD protocol}
\end{figure}

\subsection{Security Bound to the Authentication-with-response Scheme}

The crucial difference from the straightforward approach is that the security proof of~\cref{thrm:main} allows us to treat the variant as a single step that composes with the quantum key exchange.
Even though mutual authentication still consists of two parts, i.e., the $\varepsilon$-ASU\textsubscript{2} authentication and the key response, the proof analyzes them jointly\footnote{Note that the same cannot be done a priori in the straightforward approach, since the two authentication steps are performed using different keys.}, avoiding the aforementioned factor of two.
Therefore, combining an $\varepsilon_1$-secure quantum key exchange with an $\varepsilon_3$-secure authentication-with-response step provides an $(\varepsilon_1+\varepsilon_3)$-secure (mutually authenticated) QKD protocol, as depicted in~\cref{fig:multi-round-auth-with-resp}.

Similarly to the straightforward approach, we define:
\begin{itemize}
	\item $\varepsilon_{3,i}$ the $\varepsilon_3$ corresponding to the $i$-th QKD round;
	\item $\varepsilon''_i$ the perfectness parameter of the authentication key used in the $i$-th QKD round;
	\item $\hat{\varepsilon}_i$ the security parameter of the overall $i$-th QKD round;
	\item $\dot{\varepsilon}$ as the sum $\varepsilon+\frac{|\calT|}{|\calK|}$, where $\calT$ and $\calK$ are the space of tags and keys respectively related to the involved set of $\varepsilon$-ASU\textsubscript{2} functions.
\end{itemize}

Moreover, the following relations hold:
\begin{equation*}
	\varepsilon''_i = \begin{cases}
		0                       & \text{for }i = 1 \\
		\hat{\varepsilon}_{i-1} & \text{for }i>1
	\end{cases},
	\qquad \varepsilon_{3,i} = \dot{\varepsilon} + \varepsilon''_i,
	\qquad \tilde{\varepsilon}_i = \varepsilon_1 + \varepsilon_{3,i}.
\end{equation*}

Again, the general formula regarding the $i$-th round in terms of $\varepsilon$ and $\dot{\varepsilon}$ can be proved by induction. We obtain the following:
\begin{equation*}
	\varepsilon''_i=(i-1)\varepsilon_1+(i-1)\dot{\varepsilon}, \qquad
	\varepsilon_{3,i}=(i-1)\varepsilon_1+i\dot{\varepsilon}, \qquad
	\hat{\varepsilon}_i=i\varepsilon_1+i\dot{\varepsilon}.
\end{equation*}

We conclude that the provided security bound $\hat{\varepsilon}_i$ grows linearly in $i$. This proves the security of the QKD protocol, including the authentication-with-response, even after multiple rounds.

\end{document}